\documentclass[conference]{IEEEtran} 

\usepackage{bbm}
\usepackage{graphicx}
\usepackage{subfigure}
\usepackage{color}
\usepackage{epsfig}
\usepackage{epstopdf}
\usepackage{amssymb}
\usepackage[fleqn]{amsmath}
\usepackage{amsthm}
\usepackage{latexsym}
\usepackage{algorithm}
\usepackage{algpseudocode}
\usepackage{pifont}
\usepackage{tabu}
\usepackage{caption}
\newtheorem{thm}{Theorem}

\begin{document}
%
\title{Polar Codes over Fading Channels with Power and Delay Constraints}

\author{

\IEEEauthorblockN{Deekshith P. K.}
\IEEEauthorblockA{Dept. of ECE,\\Indian Institute of Science,\\ Bangalore, India.\\Email: deeks@ece.iisc.ernet.in}

\and

\IEEEauthorblockN{K. R. Sahasranand}
\IEEEauthorblockA{Dept. of ECE,\\Indian Institute of Science,\\ Bangalore, India.\\Email: sanandkr@ece.iisc.ernet.in}

}

\maketitle

\begin{abstract}
The inherent nature of polar codes being channel specific makes it difficult to use them in a setting where the communication channel changes with time. In particular, to be able to use polar codes in a wireless scenario, varying attenuation due to fading needs to be mitigated. To the best of our knowledge, there has been no comprehensive work in this direction thus far. In this work, a practical scheme involving channel inversion with the knowledge of the channel state at the transmitter, is proposed. An additional practical constraint on the permissible average and peak power is imposed, which in turn makes the channel equivalent to an additive white Gaussian noise (AWGN) channel cascaded with an erasure channel. It is shown that the constructed polar code could be made to achieve the symmetric capacity of this channel. Further, a means to compute the optimal design rate of the polar code for a given power constraint is also discussed.
\end{abstract}

\IEEEpeerreviewmaketitle

\section{Introduction}
The polar coding technique introduced by Erdal Ar{\i}kan~\cite{arikan2009channel}, was the result of his efforts to devise a practical scheme to boost the \textit{channel cutoff rate} in sequential decoding of convolutional codes. Analytical error performance guarantees, low encoding and decoding complexity and, particularly, the provable symmetric capacity (originally, of binary input memoryless channels) achieving characteristic, constitute the reason for the vogue that polar codes are. Forward error correction mechanisms have to be scalable in terms of bandwidth, data rates, energy efficiency and complexity so as to pave the way into next generation standards. Polar code is one contender in this regard~\cite{Huawei}. 

The innate nature of wireless medium offers challenges in using polar codes for communication over fading channels. This is because the channel changes with time and polar code is intrinsically channel specific. In this work, we provide a technique for using polar codes over wireless channels with delay and power constraints. So as to outline our contribution in its proper context, we first provide an overview of related literature.

In his original work~\cite{arikan2009channel}, Ar{\i}kan  proposed the idea of \textit{channel polarization} as a method to design codes that achieve the symmetric capacity of a binary input discrete memoryless channel (B-DMC). \c{S}a\c{s}\u{o}glu et al.~\cite{telatar2009polarization} extended the result to arbitrary discrete memoryless channels. An early investigation on polar coding technique achieving the capacity of an additive white Gaussian noise (AWGN) channel was undertaken in~\cite{abbe2012polar}. Taking into consideration the channel specific nature of polar codes, in a recent study~\cite{vangala2015comparative}, the authors propose a heuristic algorithm to obtain an optimized \emph{design signal to noise ratio} (design-SNR) for a given polar code construction. This is done so as to combat the performance degradation caused due to operating the codes at  SNRs different from their design SNRs.

In case of fading channels,~\cite{si2014polar} proposes a polar coding scheme for two block fading channel models namely, a binary symmetric channel with crossover probability being a finite mixture of Bernoulli distributions and an additive channel constrained by a non negative, finite mean input and noise probability density, a finite mixture of exponential distributions.  The scheme is shown to achieve the corresponding symmetric capacities of the channels, without the knowledge of instantaneous channel state at the transmitter. In~\cite{boutros2013polarization}, the authors study the use of polar codes in conjunction with \emph{channel multiplexers} for binary input block fading channels with two fading states. For a Rayleigh fading channel model, when channel state information (CSI) is available at both the transmitter (referred to as CSIT) and the receiver (CSIR),~\cite{bravo2013polar} discusses a polar coding technique by quantizing the fading states and using multiple polar codebooks. Monte Carlo approach suggested in~\cite{arikan2009channel} is used to obtain the \textit{information set} in their construction. In contrast,~\cite{trifonov2015design} proposes to model the sub channels obtained by the polar transformation of a Rayleigh fading channel by yet another fading channel with additive Gaussian noise and multiplicative noise with $\chi$ distribution. Error probabilities of sub channels computed for this model are used in code construction. 
\par In a recent arxiv submission~\cite{wasserman2016ber}, the authors compare the bit error rate performance of polar codes with that of convolutional codes for Orthogonal Frequency Division Multiplexing (OFDM) systems. Another recently reported work~\cite{tavildar2016interleaver} studies the use of polar codes along with bit interleavers (so as to improve the diversity gain) over  a block fading channel with two blocks. A study on the use of polar lattices for non binary input fading channel with independent and identically distributed (i.i.d.) fading states is reported in~\cite{liu2016polar}. A coded modulation approach for designing polar codes for block fading channels is taken up in~\cite{zheng2017polar}.  A summary of existing schemes which implement polar codes over fading channels is provided in Table~\ref{table}. 
\par All of available research could be binned into two broad categories. The first includes studies that propose polar coding schemes to use fading to their advantage, i.e., to improve diversity gain \cite{boutros2013polarization}, \cite{wasserman2016ber}, \cite{tavildar2016interleaver} at the cost of delay in communication. Studies in the other category treat fading as an impediment and  propose ways to exclusively construct polar codes for fading channels  \cite{si2014polar}, \cite{trifonov2015design}, \cite{bravo2013polar}, \cite{liu2016polar} and \cite{zheng2017polar} often times leading to complex designs. In this work, we have undertaken an alternative approach.
\par We propose a simple, practical scheme based on \emph{channel inversion}, to use polar codes in a fading environment while obeying delay constraints imposed by the data and power constraints imposed by the circuitry. The method has the advantage that a single polar codebook suffices for communication, irrespective of the channel conditions. In addition, the codebook can be constructed based on the construction for an AWGN channel without fading. Surprisingly, despite its simplicity, ease of implementation and amenability to analysis, channel inversion as a technique to use polar codes over fading channels has not been studied thus far.

\begin{table}
\tabulinesep=1.8mm
\begin{tabu}{ | c | c | c | c | c | }
\hline
 Reference & Fading & CSIT, CSIR & Metric \\ 
 \hline
 \hline
 H. Si et al.~\cite{si2014polar} & Block  & \textcolor{red}{X}, \textcolor{green}{\checkmark}  & $C_{\text{ergodic}}$\\ 
 \hline
 J. Boutros et al.~\cite{boutros2013polarization}  & Quasi static & \textcolor{green}{\checkmark}, \textcolor{green}{\checkmark} & $C_{\text{outage}}$ \\
 \hline
 A. B-Santos~\cite{bravo2013polar} & Block & \textcolor{green}{\checkmark}, \textcolor{green}{\checkmark}  &  -- \\ 
 \hline 
 P. Trifonov \cite{trifonov2015design} & Fast & \textcolor{red}{X}, \textcolor{green}{\checkmark}  & -- \\ 
 \hline
 S. R. Tavildar~\cite{tavildar2016interleaver} & Quasi static & \textcolor{red}{X}, \textcolor{green}{\checkmark}  &  $C_{\text{outage}}$  \\ 
 \hline
 L. Liu et al.~\cite{liu2016polar} & Fast & \textcolor{red}{X}, \textcolor{green}{\checkmark}  &  $C_{\text{ergodic}}$ \\ 
 \hline
 M. Zheng et al.~\cite{zheng2017polar} & Block & \textcolor{red}{X}, \textcolor{red}{X} & $C_{\text{ergodic}}$ \\ 
 \hline
\end{tabu}
\caption{A summary of existing schemes which implement polar codes over fading channels is given. Here, $C_{\text{ergodic}}$ and $C_{\text{outage}}$ refer to the ergodic and outage capacities respectively. Checkmarks and crossmarks indicate the presence and absence of CSI, respectively.}
\label{table}
\vspace{-0.3cm}
\end{table}

The paper is organized as follows. Notation and preliminaries are provided in Section~\ref{prelims}. In Section~\ref{AWGN_prelims}, we briefly describe a polar code construction method for an AWGN channel. The challenge in  extending the same to a fading channel is discussed in Section~\ref{eg_prob_dfn}. The main result of the paper is provided in Section~\ref{algo_main_result} and further expounded in Section~\ref{eq_view} by considering an equivalent model.  In Section \ref{optm}, we discuss the rate optimal design for a given power constraint and conclude with some future directions in Section~\ref{conclusion}. 

\section{Preliminaries}
\label{prelims}
A polar code is fully specified by the triple $(N,K,\mathcal{F})$ where $N = 2^n$ is the length of the codeword in bits, $K$ is the number of information bits (thereby making the rate of the code, $\frac{K}{N}$) and $\mathcal{F}$ is the set of indices in $\{1,...,N\}$ which carry $\emph{information}$ bits. In our context, it is worthwhile to view $N$ as the number of channel uses.

\emph{Encoding - } Consider two bits $u_1$ and $u_2$ which are to be encoded into bits $x_1$ and $x_2$ to be sent over a channel $W$. This is performed as shown in Figure~\ref{fig:enc}. Equivalently,
\begin{equation} \label{eqn:enc}
\hspace{.5cm}
\left[
  \begin{array}{cc}
    x_1 & x_2   
  \end{array}
\right]
   =
\left[
  \begin{array}{cc}
    u_1 &  u_2   \\
   \end{array}
\right]
\left[
  \begin{array}{cccc}
    1 & 0\\
    1 & 1\\
  \end{array}
\right],
\end{equation}
\begin{figure}[!h]
\begin{center}
\centering
\includegraphics[scale=0.55]{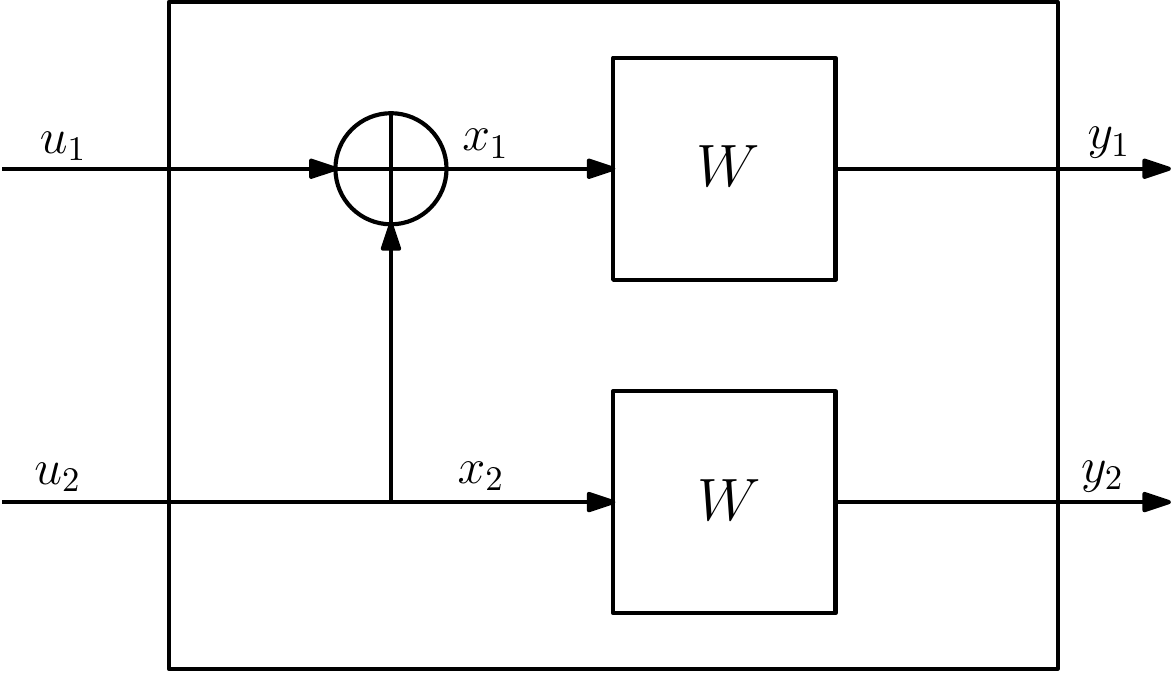}
\caption{Encoding the bits $u_1$ and $u_2$} \label{fig:enc}
\end{center}
\end{figure}
or more concisely as $x = u.F$. 

In an $(N,K,\mathcal{F})$ polar code, this procedure is repeated in a recursive fashion to generate a codeword of length $N$. The codeword $x \in \{0,1\}^N$ is obtained by multiplying the input vector $u \in \{0,1\}^N$ with the matrix $F^{\otimes n}$ where $F^{\otimes n}$ is the $n$-fold Kronecker product of the basic encoding matrix $F$ shown in Equation (\hspace{-.12cm}~\ref{eqn:enc}). The input vector $u$ is constituted of $K$ information bits and the rest of the bits \emph{frozen a priori} (set to zero). The locations of information bits is dictated by the set $\mathcal{F}$ which is calculated via the Bhattacharyya parameter of the channel. Thus, the codeword transmitted is
\begin{equation}
\hspace{2.5cm}
x = u\cdot F^{\otimes n}.
\label{eqn:codeword}
\end{equation}

\emph{Decoding - } Decoding may be performed in a number of ways such as successive cancellation (SC) or list decoding~\cite{tal2015list}, to name a few. The complexity of an SC decoder is of the order of $O(N \log N)$ where $N$ is the codeword length.

\emph{Error performance - } It was shown by Ar{\i}kan~\cite{arikan2009channel} that for a B-DMC $W$, for a polar code of length $N$ and rate $R < I(W)$ where $I(W)$ is the symmetric capacity of the channel, the bit error rate is upper bounded by $2^{-N^{\beta}}$ for some fixed $\beta < 1/2$.

\section{Designing polar codes for AWGN channel}
\label{AWGN_prelims}
In this section, we provide a brief overview of one particular construction of polar code for an AWGN channel, which we will use subsequently in the construction of polar codes for fading channels. There are various techniques for constructing polar codes for an AWGN channel. Some prominent techniques include Monte Carlo approach and polarized channel ordering scheme based on upper bounds for Bhattacharyya parameters of synthesized channels, both due to Ar{\i}kan ($\!$~\cite{arikan2009channel},~\cite{arikan2008performance}), an algorithm which estimates the transition probability matrices and hence the bit error rates (BERs) of the synthesized channels, due to Tal and Vardy~\cite{tal2013construct}, a Gaussian approximation scheme due to Trifonov~\cite{trifonov2012efficient} and a quantization scheme proposed in~\cite{abbe2012polar}. A comparative study of all the above schemes except~\cite{abbe2012polar} is undertaken in~\cite{vangala2015comparative}. It is observed therein that all the schemes perform equally well subject to choosing the design-SNR appropriately. In this work, we use the scheme proposed originally by Ar{\i}kan~\cite{arikan2008performance} incorporating the modification mentioned in~\cite{vangala2015comparative}. In what follows, we provide some critical details of this construction. 

To begin with, let us examine the channel for which the polar code is to be constructed. The codeword $x=(x_1,\hdots,x_N) $ is to be modulated, symbol by symbol, using a Binary Phase Shift Keying (BPSK) scheme with power $P$ to obtain $X =\big(X_1,\hdots,X_N\big) \in \{{\pm \sqrt{P}}\}^N$ which is transmitted over an AWGN channel. This is represented as $$Y_k = X_k + \eta_k,~k=1,2,\ldots,N,$$  where $\{\eta_k:1 \le k \le N\}$ is i.i.d. with common distribution, Gaussian with mean $0$ and variance $\sigma^2$, denoted $\mathcal{N}(0,\sigma^2)$.

Fix a rate $\frac{K}{N}= R \in [0,1]$. Obtain the BPSK signalling power $P$ required to achieve this rate over a binary input AWGN channel with noise variance $\sigma^2$ by solving the following equation: 
\begin{equation}
R= -\int_{-\infty}^{+\infty}f_Y(y)\log\big(f_Y(y)\big)dy -\frac{1}{2}\log(2\pi e\sigma^2),
\label{eqn:bin}
\end{equation}
where $$f_Y(y)=\frac{1}{2}\frac{1}{\sqrt{2\pi\sigma^2}}e^{-\frac{(y-\sqrt{P})^2}{2\sigma^2}}+\frac{1}{2}\frac{1}{\sqrt{2\pi\sigma^2}}e^{-\frac{(y+\sqrt{P})^2}{2\sigma^2}}.$$ The design-SNR for the construction of polar code is thus fixed to be $\frac{P}{\sigma^2}$ (denoted from now on as, simply SNR). This is the sole parameter required to obtain the information set $\mathcal{F}$. Next, we set out to sketch the details of the same. 

\par For a given channel with binary input alphabet $\{\pm \sqrt{P}\}$ and transition probability density function $W(\cdot|\cdot)$, the Bhattacharyya parameter of the channel is defined to be $$Z(W) \triangleq \int_{-\infty}^{+\infty} \sqrt{ W\big(y\lvert-\sqrt{P}\big)W\big(y\lvert+\sqrt{P}\big)} dy.$$ Accordingly, for an AWGN channel with BPSK input, we obtain the Bhattacharyya parameter $Z(SNR)=e^{-SNR}$. 

\par For the given AWGN channel, let $Z_{i,j}(SNR)$ denote the Bhattacharyya parameter of the $j^{\text{th}}$ polarized channel, synthesized after $i$ polar transformations of the original channel, where $j \in \{1,\hdots,2^i\}$ and $i \in \{0,\hdots,n\}$. Then, define the following recursive relationship as in \cite{arikan2008performance}:  
  \begin{equation}
  \hspace{.65cm}
  \begin{array}{lr}
Z_{i+1,2j-1}(SNR) = Z_{i,j}^2(SNR),\\
Z_{i+1,2j}(SNR) =  2Z_{i,j}-Z_{i,j}^2(SNR).
\end{array}
\label{eqn:bhat}
\end{equation}
The above recursion is initialized to $Z_{0,1}(SNR)=e^{-SNR}$. Using the recursion, we obtain the vector $v_Z=\big(Z_{n,1}(SNR),\hdots,~Z_{n,N}(SNR)\big)$. Let $\pi_N$ be a permutation function and $\pi_N(v_Z)=\big(Z_{n,\pi_1}(SNR),\hdots,~Z_{n,\pi_N}(SNR)\big)$ such that $Z_{n,\pi_l} \leq Z_{n,\pi_m}$ for $1 \leq l <m \leq N$. Obtain the information vector $\mathcal{F}$ by choosing the first $K$ elements of $\pi_N(v_Z)$. This completes the discussion of polar code construction for AWGN channels without fading.
\section{Channel Inversion }
\label{eg_prob_dfn}
\par The optimality of polar codes (in the sense of achieving the symmetric capacity~\cite{arikan2009channel}) stems from the polarization of \emph{effective} channels of the bits $u_k$s, into good and bad channels. However, it may be noted that the \emph{actual} channel that the codeword bits $x_k$s are transmitted through, is (statistically) the same $W$ for all $x_k,~ k=1,\hdots,N$. The presence of multipath fading renders different, the channel each of the transmitted symbols $X_k$, sees. In particular, the channel output is$$Y_k = H_kX_k + \eta_k,~k=1,2,\ldots,N,$$ where $\{H_k:1 \le k \le N\}$ are i.i.d. channel gains with common distribution $F_H$. 
\par When the instantaneous channel gains $H_k$s are available at the transmitter (CSIT) before the transmission of $X_k$, it is possible to instead transmit $\tilde{X}_k=\frac{X_k}{H_k}$ (from a BPSK constellation with power $\frac{P}{H_k^2}$) to \emph{nullify} the effect of the channel gain. This is known as \emph{channel inversion}. This ensures that the effective channel remains statistically the same for all transmissions. However, to perform the inversion, the expected amount of energy expended is $P\cdot\mathbb{E}[\frac{1}{H^2}]$ which can be prohibitively large. For  example, when $F_H \equiv \mathcal{N}(0,\sigma_H^2),~P\cdot\mathbb{E}[\frac{1}{H^2}]=\infty$. But  practically, there is only a finite amount of power available for transmission. This is modelled by imposing an average power constraint $\mathbb{E}[\zeta^2]\leq Q<\infty$ at the transmitter, where  $\zeta$ is the transmitted symbol.
\section{Design with power constraint: Fading Channel}
\label{algo_main_result}
 In a practical scenario, although the polar code is designed at a \emph{design-SNR}, $\frac{P}{\sigma^2}$, it may be operated at a different \emph{operating-SNR}, say $\frac{Q}{\sigma^2}$, where $Q$ is the average power constraint (we use the terms operating power and average power constraint synonymously hereafter). See~\cite{vangala2015comparative} for a detailed study in this direction. For an AWGN channel, this mismatch in design and operating SNRs translates to a change in error performance for the same rate at which the polar code was originally designed. However, in a fading scenario, the power constraint $Q$ dictates the (reduced) effective rate at which symbols are transmitted. This is because, it is impossible to invert all the channels the codeword sees, for that would require an infinite amount of power as observed in the previous section. This results in the loss of codeword symbols at those instances when the channel is too bad to be inverted.

\par When the channel gain is $H_k$, the symbol transmitted is chosen from a BPSK constellation whose power is $P\cdot\frac{1}{H_k^2}$. Note that when $|H_k|$ is close to zero, a huge amount of power is required to invert the channel. This suggests the following strategy. When $|H_k|$ is less than a certain $\overline{\delta}$, transmit at \emph{zero} power (i.e., transmit nothing). During all other instances, transmit at power $P\cdot\frac{1}{H_k^2}$. This results in some codeword symbols getting lost, whenever the channel is bad \big($|H_k| < \overline{\delta}$\big). Thus, we would like $\overline{\delta}$ to be as small as possible. This, however, is constrained by the average amount of power available, $Q < \infty$. We calculate $\overline{\delta}$ as follows:
\begin{equation}
\int_{-\infty}^{-\overline{\delta}} P\cdot\frac{1}{h^2}\cdot dF_H + \int_{\overline{\delta}}^{\infty} P\cdot\frac{1}{h^2}\cdot dF_H  = Q.
\label{eqn:avgP}
\end{equation}
In addition to the average power constraint, owing to practical considerations such as the need to control the Peak to Average Power Ratio (PAPR), as high PAPR signals usually strain the analog circuitry, a peak power constraint is also imposed. Suppose that the transmit power is constrained to be always less than $\tilde{Q}$. That is, we perform channel inversion only when $\frac{P}{H_k^2} \le \tilde{Q}$. Equivalently, when
\begin{align}
\hspace{2.5cm}
 |H_k| \ge \sqrt{\frac{P}{\tilde{Q}}}.
\label{eqn:peakP}
\end{align}

Thus, under a setting with average power constraint $Q$ and peak power constraint $\tilde{Q}$, the channel is not inverted during time instant $k$, when $|H_k|<\delta$ where
\begin{equation}
\hspace{2.1cm}
\delta = \max \Bigg\{\overline{\delta}, \sqrt{\frac{P}{\tilde{Q}}}\Bigg\}.
\label{eqn:delta}
\end{equation}
This is called \textit{truncated channel inversion} and is known to be optimal for delay constrained communication~\cite{yang2015optimum}. The algorithm for transmitting data using polar codes over a fading channel with CSIT is given in two parts (see Figure~\ref{block_fad_inv} for a schematic). The first part, shown below in Algorithm ($\!\!$~\ref{algo:polaroff}), can be executed offline, before the actual transmission of symbols. The second part, Algorithm ($\!\!$~\ref{algo:polaron}) is online and is executed for each channel use.
\label{eqcha}
\begin{figure}[h]
\begin{center}
\centering
\includegraphics[scale=0.55]{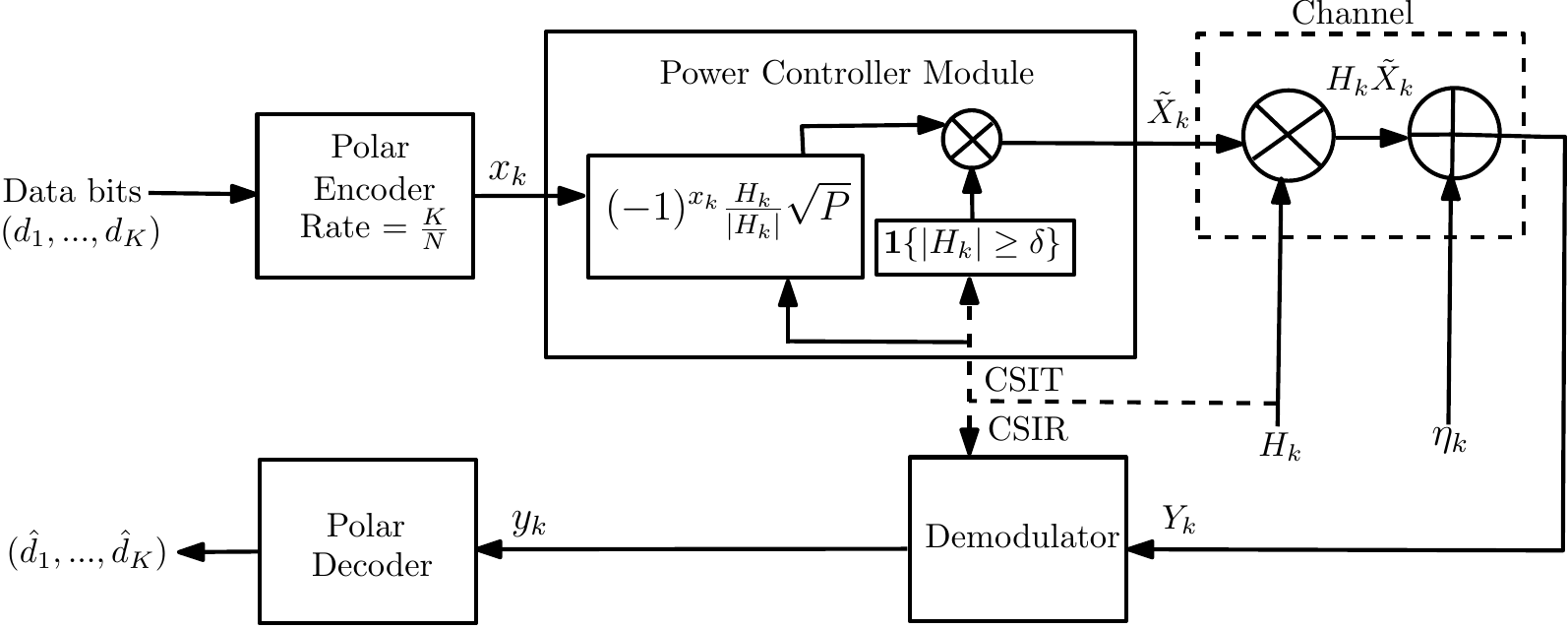}
\caption{Block diagram of the polar coded communication system.} \label{block_fad_inv}
\end{center}
\end{figure}

\begin{algorithm}[h]
\caption{Polar coded data transmission for Binary Input fading AWGN channels: Offline part}
\label{algo:polaroff}
\textbf{INPUT:} Information bits $\{d_j: 1 \le j \le K\}, N, K, Q, \tilde{Q}$, AWGN statistics, $\mathcal{N}(0,\sigma^2)$.\\
\textbf{OUTPUT:} Design power $P$, Codeword $x$, Threshold $\delta$.
\begin{algorithmic}[1]
\State Initialize $u_i=0, 1 \le i \le N$.
\State Calculate the power $P$ for which the rate over the binary input AWGN channel is $R=\frac{K}{N}$ using Equation ($\!\!$~\ref{eqn:bin}).
\State Construct the set $\mathcal{F}^C$ which consists of locations of $u$ which are to be set to zero, using Equation ($\!\!$~\ref{eqn:bhat}).
\State Copy the information bits $\{d_j:1\le j \le K\}$, one by one into $\{u_i: i \in \mathcal{F}\}$.
\State Set $x = u\cdot F^{\otimes n}$ as in Equation ($\!\!~\ref{eqn:codeword}$).
\State Calculate the threshold $\delta$ below which the channel shall not be inverted, using Equations ($\!\!~\ref{eqn:avgP}$), ($\!\!~\ref{eqn:peakP}$) and ($\!\!~\ref{eqn:delta}$).
\end{algorithmic}
\end{algorithm}
For each time instant $k$, the state of the channel, $H_k$ is made available at the transmitter. This knowledge is used in performing a channel inversion, subject to the average power constraint $Q$ and peak power constraint $\tilde{Q}$. Algorithm ($\!\!$~\ref{algo:polaron}) is run during every time instant, $1 \le k \le N$ to obtain the constellation symbol to be transmitted, $X_k$. Bit $0$ is mapped to $+\sqrt{T}$ and bit $1$ is mapped to $-\sqrt{T}$ when the BPSK constellation used is $\{\pm \sqrt{T}\}$. When $T=0$, nothing is transmitted and the corresponding codeword symbol is skipped over. The decoder, in possession of CSIR, (can compute $T$ and hence) ignores the received symbol at time instances corresponding to $T=0$. 
\begin{algorithm}
\caption{Polar coded data transmission for Binary Input fading AWGN channels: Online part}
\label{algo:polaron}
\textbf{INPUT:} Design power $P$, Codeword $x$, Threshold $\delta$, Channel gain for the instant $H_k$.\\
\textbf{OUTPUT:} Transmit power $T$, Constellation symbol $\tilde{X}_k \in \{\pm \sqrt{T},0\}$.
\begin{algorithmic}[1]
\State Initialize $T=0$.
\State If $|H_k| \ge \delta$, set $T=\frac{P}{H_k^2}$.
\State Map $\tilde{X}_k = \frac{H_k}{|H_k|}\cdot(-1)^{x_k}\cdot \sqrt{T}$.
\end{algorithmic}
\end{algorithm}

\section{Design with power constraint: An Equivalent View Point}
\label{eq_view}
We described our algorithm for using polar codes over a fading channel in the previous section. In this section, we subscribe to an abstract but nevertheless beneficial alternative view point of the same. 
\par As in the previous section, consider a fading AWGN channel with channel state information known at both the transmitter and the receiver. Given a polar codebook designed at a target rate $R$ with design-SNR $\frac{P}{\sigma^2}$, permissible average power $Q$ and peak power $\tilde{Q}$, we can obtain $\delta$ as in Equation ($\!\!$~\ref{eqn:delta}) so that the channel is inverted only in those time instances $k$ when $|H_k| \ge \delta$. This would result in an $\epsilon$ fraction of symbols getting lost where $\epsilon$ is calculated as
\begin{equation}
\hspace{2.5cm}
\epsilon = \int_{-\delta}^{\delta}dF_H.
\label{eqn:epsilon}
\end{equation}
The event $\{|H_k|<\delta\}$ corresponds to a \emph{deep fade} for small $\delta$ (the probability that the channel is in deep fade during a time slot is called \emph{outage} probability in wireless parlance). The dependence of $\epsilon$, the outage probability on the operating power $Q$ for the fading model in Section~\ref{eg_prob_dfn}, is illustrated in Figure \ref{fig_eps_vs_Q} for different values of $\sigma_H^2$.

\begin{figure}[h]
\begin{center}
\centering
\includegraphics[scale=0.4]{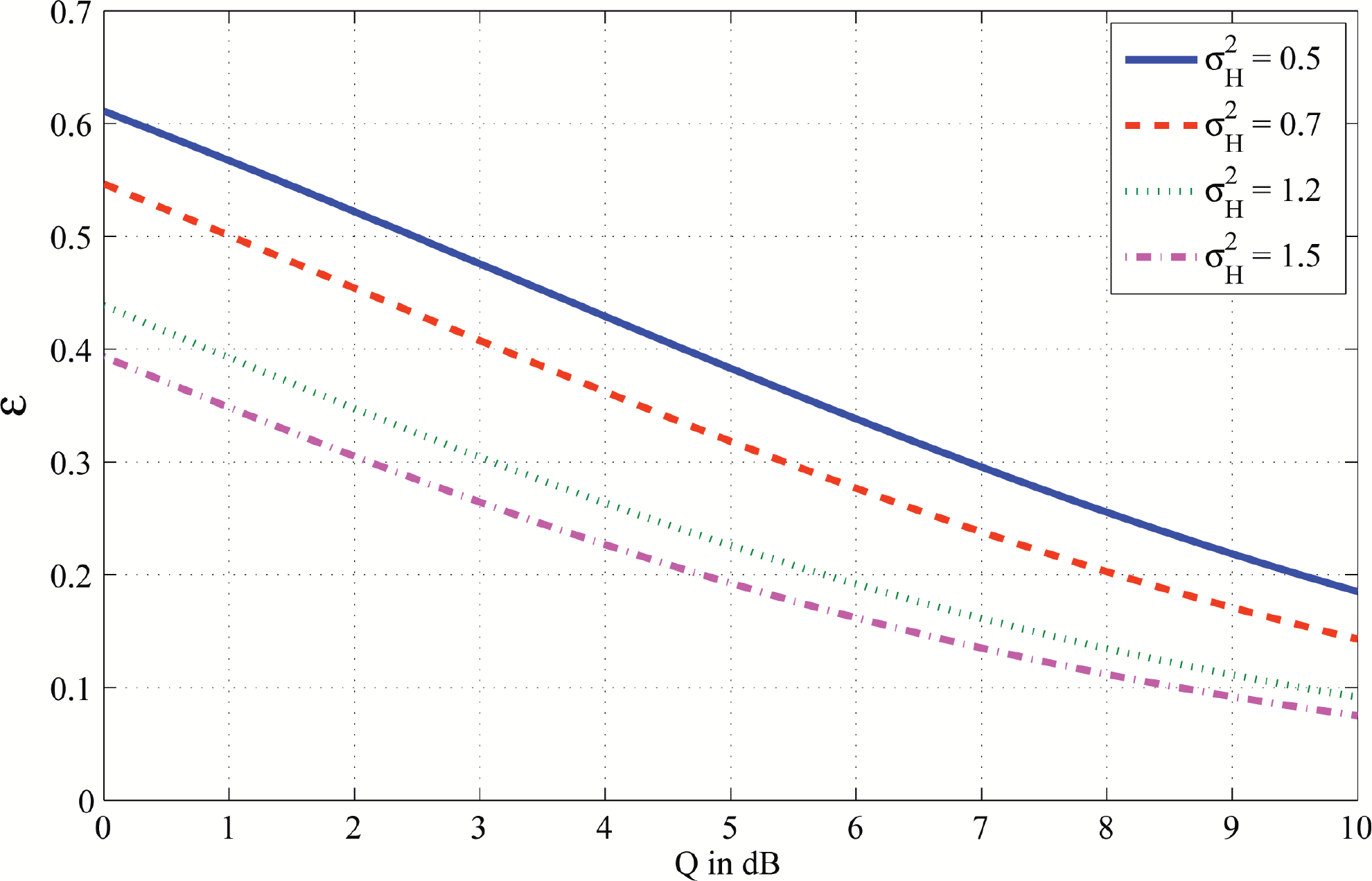}
\caption{Dependence of erasure probability on operating power is shown. The parameter values are $\sigma^2 = 1,~ R = 0.5$ and $P$ is accordingly calculated.} 
\label{fig_eps_vs_Q}
\end{center}
\end{figure}
\par In any time slot, there is a probability $(1-\epsilon)$ that, the transmitter \emph{inverts} the channel. Let $C_{AWGN}$ denote the capacity (or more precisely the binary input symmetric capacity, as we are considering only binary input channels) of the AWGN channel with input SNR $\frac{P}{\sigma^2}$ and $C_{outage}$ be the (binary input symmetric) outage capacity (\cite{tse2005fundamentals}) of the fading channel. If the channel is in outage, the transmitter does not transmit any symbol. Since CSIR is  available, the receiver can keep track of the time instances at which the transmissions happen and discard the channel outputs at the remaining time instants. Equivalently, time instants at which channel outputs are discarded can be considered as erasures introduced by the channel. 
\par Thus, we can think of a fading channel with CSIT and CSIR available, over which the transmitter sends data at a constant rate (via channel inversion) during \textit{most} of the time and remains \textit{silent} at others, as the cascade of an AWGN channel and an erasure channel (see Figure~\ref{eq_cha}). Polar codes can be designed so as to achieve the binary input symmetric capacity of this equivalent channel. Note that, \textit{erasures} in the equivalent channel occur independently of the output of the AWGN channel. Hence, the symmetric  capacity of the equivalent channel is given by 
$C_{eq} = (1-\epsilon)C_{AWGN}.$
This, by equivalence, is also the outage capacity $C_{outage}$.
\begin{figure}[h]
\begin{center}
\centering
\includegraphics[scale=01]{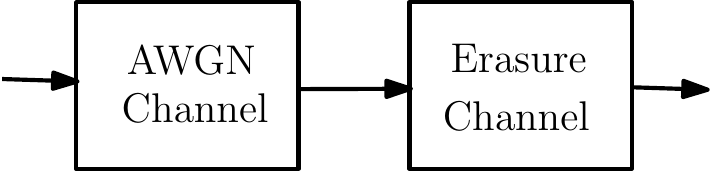}
\caption{An equivalent channel model} 
\label{eq_cha}
\end{center}
\end{figure}
\begin{thm}
The polar coded communication scheme constituted by Algorithm $1$ and Algorithm $2$ achieves the binary input symmetric outage capacity of a fading  AWGN channel for an arbitrarily large blocklength $N$.
\end{thm} 
\begin{proof}
The equivalent channel (the cascade of binary input AWGN channel and erasure channel) is a B-DMC. Note that our scheme uses the symmetric capacity achieving polar code construction for the underlying (time invariant) AWGN channel \cite{arikan2009channel}, \cite{arikan2008performance}. A code of rate $R<C_{AWGN}$ designed for the AWGN channel, operates at a rate $(1-\epsilon)R$  under our scheme, when used over the equivalent channel. For  large $N$, a polar code of rate $R$ arbitrarily close to $C_{AWGN}$ can be constructed. The same code, using our scheme, operates at a rate arbitrarily close to $(1-\epsilon)C_{AWGN}=C_{eq}=C_{outage}$. 
\end{proof} 
There is a subtle, but important point to be noted while designing polar codes for the equivalent channel and using it over the fading channel. While designing the polar codes, the channel states are not known a priori. For the same reason, say, the code construction is done for the \textit{average} channel, i.e., for the mixture channel $(1-\epsilon)W(\cdot|\cdot)+ \epsilon W(e|\cdot)$, where $e$ denotes the erasure symbol. However in practice, the transmitter and the receiver can track the instantaneous fading states and hence the channel is \textit{not a mixture channel}. The channel capacity of the \textit{average} channel is greater than or equal to that of the equivalent channel. This is due to the concavity of the channel capacity, which is a mutual information functional, as a function of the channel input distribution. Hence, polar codes designed for the mixture channel will be operating above the outage capacity of the fading channel, and the bit error probability performance will be slightly inferior as illustrated in Figure \ref{fig_compare}.
\begin{figure}[h]
\begin{center}
\centering
\includegraphics[scale=0.6]{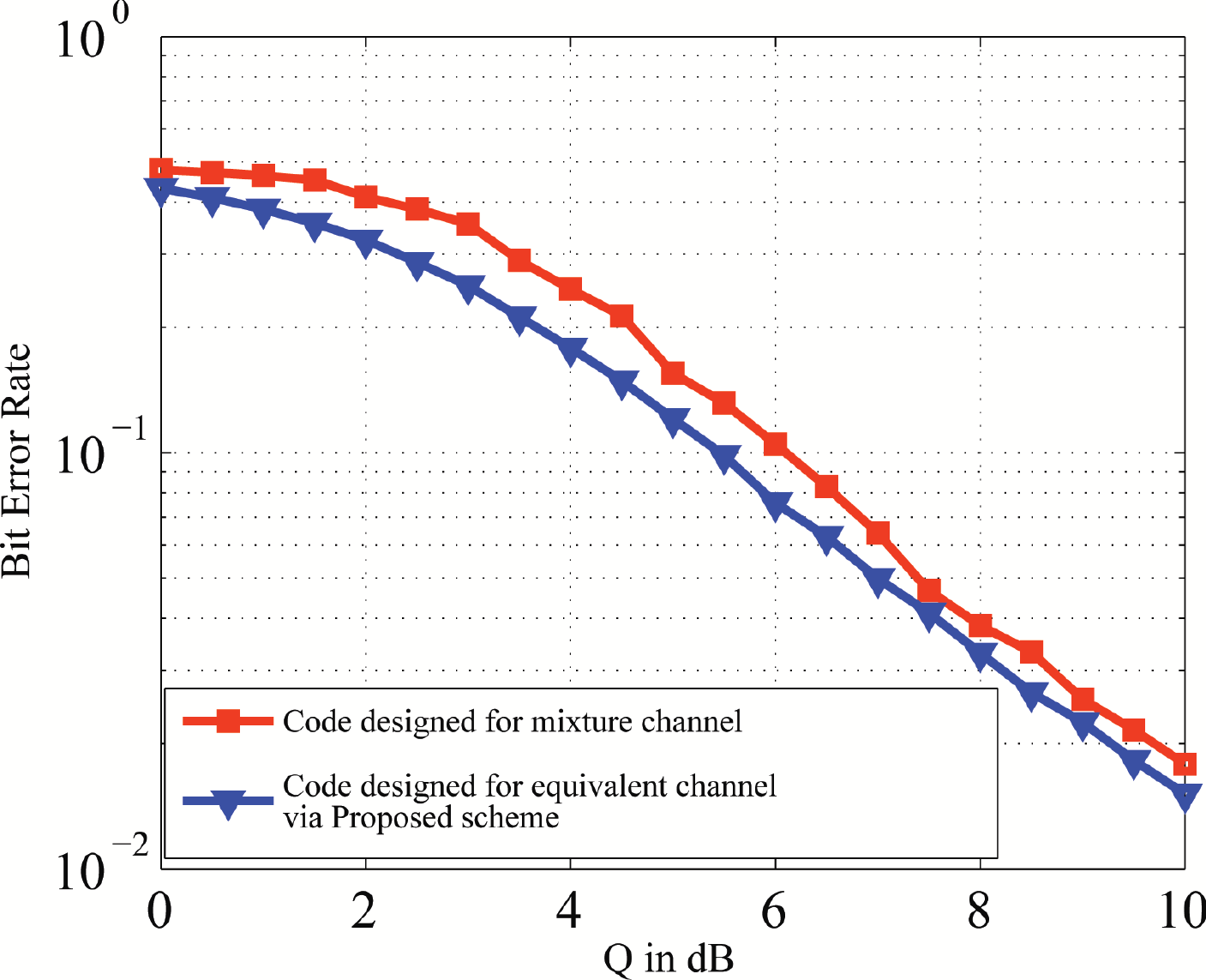}
\caption{Bit error rates for polar code construction using our scheme and polar code designed for the mixture channel. The parameters are $\sigma^2=1, F_H \equiv \mathcal{N}(0,1), R=0.5$ and $P$ calculated accordingly, and $N=1024$.} 
\label{fig_compare}
\end{center}
\end{figure}
\section{Rate optimal design for a given operating power}
\label{optm}
\par It is clear that for a fixed rate $R$ (and hence, a fixed design power $P$), since erasure probability $\epsilon$ monotonically decreases  with the average operating power $Q$, \emph{the best possible equivalent channel} corresponds to $Q$  being arbitrarily large. Alternatively, for a fixed $Q$, is there a \emph{rate optimal} design SNR $\frac{P}{\sigma^2}$ at which to construct the polar code? The optimal design power is obtained as a solution to the following optimization problem:$$\max_P ~(1-\epsilon) \Big[-\int_{-\infty}^{+\infty}f_Y(y)\log\big(f_Y(y)\big)dy \Big], $$ subject to Equations (\ref{eqn:avgP}) and (\ref{eqn:epsilon}), where $$f_Y(y)=\frac{1}{2}\frac{1}{\sqrt{2\pi\sigma^2}}e^{-\frac{(y-\sqrt{P})^2}{2\sigma^2}}+\frac{1}{2}\frac{1}{\sqrt{2\pi\sigma^2}}e^{-\frac{(y+\sqrt{P})^2}{2\sigma^2}}.$$ Note that the additive noise variance $\sigma^2$ remains the same throughout. The resulting optimal design power $P^*$ can be used to compute the optimal rate $R^*$ as per Equation (\ref{eqn:bin}). The optimal $R^*$ is plotted against the operating power $Q$ in Figure~\ref{fig_optR} (for $\sigma^2=1$). Thus, when provided with an operating power constraint $Q$, the best polar code (in terms of rate) can be chosen accordingly and then, operated over the fading channel using our algorithm (explained in Section~\ref{algo_main_result}).
\begin{figure}[ht]
\begin{center}
\centering
\includegraphics[scale=0.55]{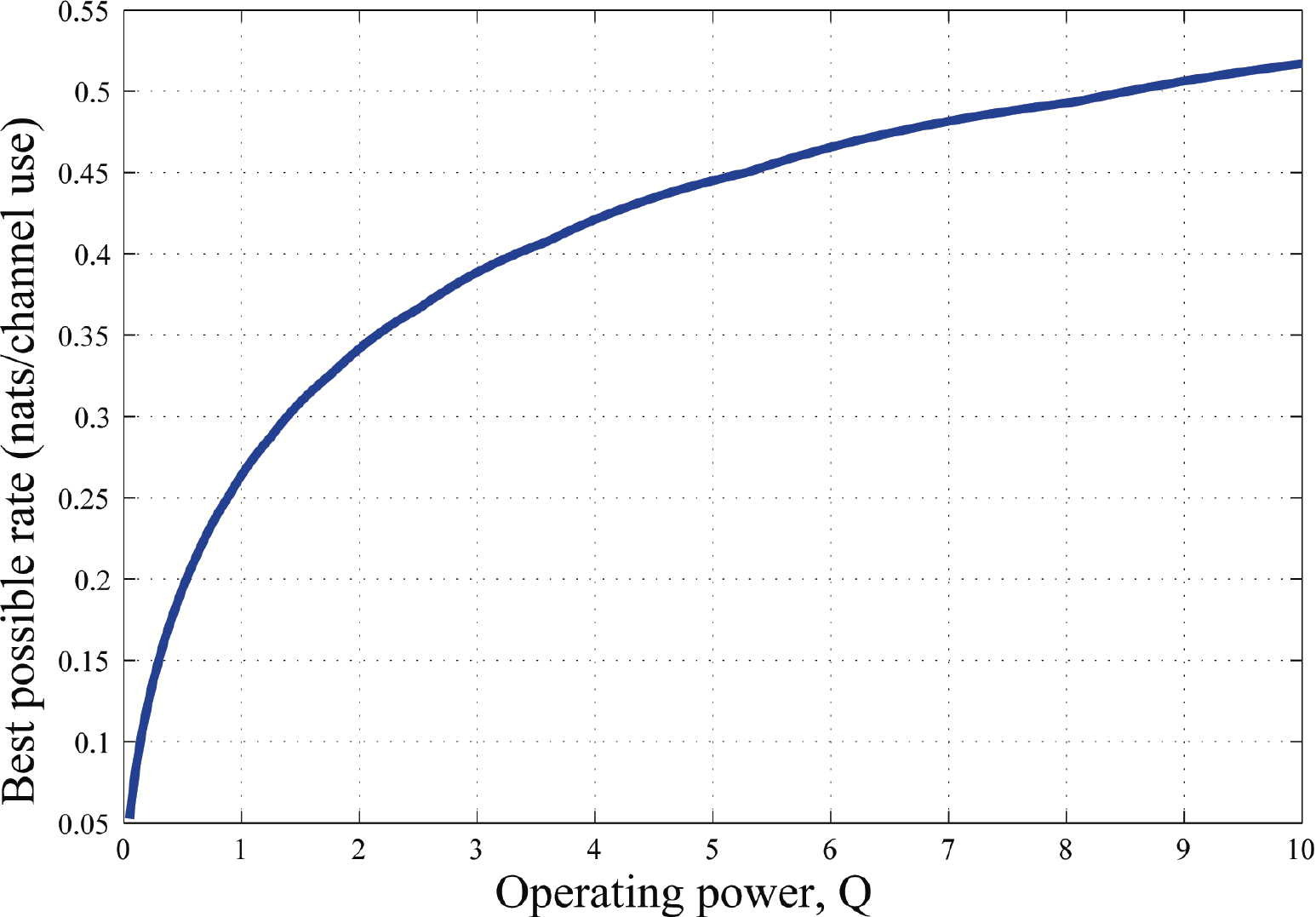}
\caption{Optimal design rate $R^*$ for a given operating power.} 
\label{fig_optR}
\end{center}
\end{figure}
 
\section{Discussion and Conclusion}
\label{conclusion}
Channel inversion as a strategy to combat fading, although power inefficient, is suitable for delay critical applications. Note that the assumption of channel gains being i.i.d. is not crucial. The inversion scheme is oblivious to the correlation between the fading states and hence works equally well irrespective of whether the fading considered is \emph{slow} or \emph{fast}. In addition, the presence of fading does not increase  the complexity of polar code construction in comparison with that of an AWGN channel under our scheme. 
\par It is clear that the proposed scheme requires CSIT to function. It is the receiver that makes CSI available to the transmitter via feedback. However, due to bandwidth constraints, a somewhat more realistic assumption is that only a \emph{quantized version} of the channel state $H_k$, say $\widehat{H}_k$, is fed back at time instant $k$. The power expended at time instant $k$ under the channel inversion scheme would then be ${P}{\widehat{H}_k^{-2}}$. Hence, the quantization scheme should be designed  such that $\mathbb{E}_{\widehat{H}}[P\widehat{H}^{-2}]\leq Q$. Further, the effect of channel state quantization on channel polarization remains to be explored.   

\section*{Acknowledgment}
The authors would like to thank Shahid Mehraj Shah, Dept. of ECE, IISc, for useful discussions.

\bibliographystyle{IEEEtran}
\bibliography{bibfile_fading_polar}

\begin{thebibliography}{10}
\providecommand{\url}[1]{#1}
\csname url@samestyle\endcsname
\providecommand{\newblock}{\relax}
\providecommand{\bibinfo}[2]{#2}
\providecommand{\BIBentrySTDinterwordspacing}{\spaceskip=0pt\relax}
\providecommand{\BIBentryALTinterwordstretchfactor}{4}
\providecommand{\BIBentryALTinterwordspacing}{\spaceskip=\fontdimen2\font plus
\BIBentryALTinterwordstretchfactor\fontdimen3\font minus
  \fontdimen4\font\relax}
\providecommand{\BIBforeignlanguage}[2]{{%
\expandafter\ifx\csname l@#1\endcsname\relax
\typeout{** WARNING: IEEEtran.bst: No hyphenation pattern has been}%
\typeout{** loaded for the language `#1'. Using the pattern for}%
\typeout{** the default language instead.}%
\else
\language=\csname l@#1\endcsname
\fi
#2}}
\providecommand{\BIBdecl}{\relax}
\BIBdecl

\bibitem{arikan2009channel}
E.~Ar{\i}kan, ``Channel polarization: A method for constructing
  capacity-achieving codes for symmetric binary-input memoryless channels,''
  \emph{IEEE Transactions on Information Theory}, vol.~55, no.~7, pp.
  3051--3073, 2009.

\bibitem{Huawei}
``5{G}: New {A}ir {I}nterface and {R}adio {A}ccess {V}irtualization,''
  \emph{Huawei White Paper}, April, 2015.

\bibitem{telatar2009polarization}
E.~Sasoglu, E.~Telatar, and E.~Ar{\i}kan, ``Polarization for arbitrary discrete
  memoryless channels,'' in \emph{Information Theory Workshop, 2009. ITW 2009.
  IEEE}.\hskip 1em plus 0.5em minus 0.4em\relax IEEE, 2009, pp. 144--148.

\bibitem{abbe2012polar}
E.~Abbe and E.~Telatar, ``Polar codes for the-user multiple access channel,''
  \emph{IEEE Transactions on Information Theory}, vol.~58, no.~8, pp.
  5437--5448, 2012.

\bibitem{vangala2015comparative}
H.~Vangala, E.~Viterbo, and Y.~Hong, ``A comparative study of polar code
  constructions for the {A}{W}{G}{N} channel,'' \emph{arXiv preprint
  arXiv:1501.02473}, 2015.

\bibitem{si2014polar}
H.~Si, O.~O. Koyluoglu, and S.~Vishwanath, ``Polar coding for fading channels:
  binary and exponential channel cases,'' \emph{IEEE Transactions on
  Communications}, vol.~62, no.~8, pp. 2638--2650, 2014.

\bibitem{boutros2013polarization}
J.~J. Boutros and E.~Biglieri, ``Polarization of quasi-static fading
  channels,'' in \emph{Information Theory Proceedings (ISIT), 2013 IEEE
  International Symposium on}.\hskip 1em plus 0.5em minus 0.4em\relax IEEE,
  2013, pp. 769--773.

\bibitem{bravo2013polar}
A.~Bravo-Santos, ``Polar codes for the {R}ayleigh fading channel,'' \emph{IEEE
  Communications Letters}, vol.~17, no.~12, pp. 2352--2355, 2013.

\bibitem{trifonov2015design}
P.~Trifonov, ``Design of polar codes for {R}ayleigh fading channel,'' in
  \emph{2015 International Symposium on Wireless Communication Systems
  (ISWCS)}.\hskip 1em plus 0.5em minus 0.4em\relax IEEE, 2015, pp. 331--335.

\bibitem{wasserman2016ber}
D.~R. Wasserman, A.~U. Ahmed, and D.~W. Chi, ``{B}{E}{R} performance of polar
  coded {O}{F}{D}{M} in multipath fading,'' \emph{arXiv preprint
  arXiv:1610.00057}, 2016.

\bibitem{tavildar2016interleaver}
S.~R. Tavildar, ``An interleaver design for polar codes over slow fading
  channels,'' \emph{arXiv preprint arXiv:1610.04924}, 2016.

\bibitem{liu2016polar}
L.~Liu and C.~Ling, ``Polar codes and polar lattices for independent fading
  channels,'' \emph{IEEE Transactions on Communications}, vol.~64, no.~12, pp.
  4923--4935, 2016.

\bibitem{zheng2017polar}
M.~Zheng, M.~Tao, W.~Chen, and C.~Ling, ``Polar coding for block fading
  channels,'' \emph{arXiv preprint arXiv:1701.06111}, 2017.

\bibitem{tal2015list}
I.~Tal and A.~Vardy, ``List decoding of polar codes,'' \emph{IEEE Transactions
  on Information Theory}, vol.~61, no.~5, pp. 2213--2226, 2015.

\bibitem{arikan2008performance}
E.~Ar{\i}kan \emph{et~al.}, ``A performance comparison of polar codes and
  {R}eed-{M}uller codes,'' \emph{IEEE Commun. Lett}, vol.~12, no.~6, pp.
  447--449, 2008.

\bibitem{tal2013construct}
I.~Tal and A.~Vardy, ``How to construct polar codes,'' \emph{IEEE Transactions
  on Information Theory}, vol.~59, no.~10, pp. 6562--6582, 2013.

\bibitem{trifonov2012efficient}
P.~Trifonov, ``Efficient design and decoding of polar codes,'' \emph{IEEE
  Transactions on Communications}, vol.~60, no.~11, pp. 3221--3227, 2012.

\bibitem{yang2015optimum}
W.~Yang, G.~Caire, G.~Durisi, and Y.~Polyanskiy, ``Optimum power control at
  finite blocklength,'' \emph{IEEE Transactions on Information Theory},
  vol.~61, no.~9, pp. 4598--4615, 2015.

\bibitem{tse2005fundamentals}
D.~Tse and P.~Viswanath, \emph{Fundamentals of wireless communication}.\hskip
  1em plus 0.5em minus 0.4em\relax Cambridge university press, 2005.

\end{thebibliography}

\end{document}